\documentclass[12pt]{ouparticle}

\usepackage{algorithm}
\usepackage{algorithmic}
\usepackage{amsmath}
\usepackage{amssymb}
\usepackage{amsthm}
\usepackage{bm}
\usepackage{bbm}
\usepackage{braket}
\usepackage{breqn}
\usepackage{enumitem}
\usepackage{hyperref}
\usepackage{mathtools}
\usepackage{scalerel}
\usepackage{stackengine,wasysym}
\usepackage{xcolor}

\newcommand{\ind}[1]{\mathbbm{1}_{\{#1\}}}

\DeclareMathOperator{\Var}{Var}
\DeclareMathOperator{\Ex}{\mathbb{E}}
\DeclareMathOperator{\p}{\mathbb{P}}

\newtheorem{definition}{Definition}
\newtheorem{lemma}{Lemma}
\newtheorem{theorem}{Theorem}

\begin{document}

\title{Differentially Private Finite Population  Estimation via Survey Weight Regularization}

\author{%
\name{Jeremy Seeman}
\address{Michigan Institute for Data Science and Institute for Social Research \\
University of Michigan, Ann Arbor, MI,\\
48104, USA}
\email{jhseeman@umich.edu}
\and
\name{Yajuan Si}
\address{Institute for Social Research \\
University of Michigan, Ann Arbor, MI,\\
48104, USA}
\email{yajuan@umich.edu}
\and
\name{Jerome P. Reiter}
\address{Department of Statistical Science \\
Duke University, Durham, NC \\
27708, USA
}
\email{jreiter@duke.edu}
}

\maketitle

\abstract{    
In general, it is challenging to release differentially private versions of survey-weighted statistics with low error for acceptable privacy loss. This is because weighted statistics from complex sample survey data can be more sensitive to individual survey response and weight values than unweighted statistics, resulting in 
differentially private mechanisms that can add substantial noise to the unbiased estimate of the finite population quantity.  On the other hand, simply disregarding the survey weights adds noise to a biased estimator, which also can result in an inaccurate estimate.  Thus, the problem of releasing an accurate survey-weighted estimate essentially involves 
a trade-off among bias, precision, and privacy. We leverage this trade-off to develop a 
differentially private method for estimating finite population quantities.  The key step is to privately estimate a hyperparameter that determines how much to regularize or shrink survey weights as a function of privacy loss. We illustrate the differentially private finite population estimation using the Panel Study of Income Dynamics.  We show that optimal strategies for releasing DP survey-weighted mean income estimates require orders-of-magnitude less 
noise than naively using the original survey weights without modification.
}

\keywords{confidentiality; disclosure; sampling design; and weighting}

\section{Introduction}

As part of efforts to protect data subjects' privacy and confidentiality, data stewards can release statistics that satisfy differential privacy (DP) \cite{dwork_calibrating_2006,dwork_algorithmic_2014}. 
 To date, typical applications of DP have been based on data from censuses or administrative databases.  Often, however, statistics are based on surveys with complex designs, e.g., using multi-stage, unequal probability sampling. It is well known that analysts should account for the design in inferences, which is typically done by using survey weights. Survey-weighted statistics offer unbiased and consistent estimates for finite population quantities, such as population means and totals. Yet, as we describe below, survey weights introduce challenges to implementing DP methods. These challenges motivate our work: how might data stewards apply DP to release survey-weighted statistics from complex sample surveys?

Preliminary work at the intersection of DP and survey statistics has focused on synthetic data generation with weights \cite{hu_private_2022}, estimation under classical sampling designs like stratified sampling \cite{lin_differentially_2023}, and interpretations of survey sampling methods for their privacy amplification properties \cite{bun_controlling_2022,hu_accuracy_2022}. Each of these approaches attempts to utilize as much information as possible about the sampling process. However, weighting schemes can cause practical problems for DP.  Weighted statistics can have significantly larger sensitivities than their unweighted counterparts, hence requiring substantially more noise to provide the equivalent level of DP protections at the same level of privacy loss \cite{drechsler_differential_2023, reiter:arisa}.  Of course, one could avoid the associated increase in the DP noise variance by disregarding the sampling weights in estimation. Indeed, this can be appealing when the weighted and unweighted estimates are similar, which can occur when the survey weights are uncorrelated with the particular survey variable of interest \cite{little-weighting-SM2005, bollen_are_2016, si2023population}. However, when this is not the case, the data steward ends up adding noise to a (perhaps severely) biased estimate. 

These two extremes  suggest that DP survey-weighted estimation involves a trade-off among bias, precision, and privacy.  Indeed, such trade-offs are common for DP estimation tasks even with independently identically (iid) data as well \cite{kamath_bias-variance-privacy_2023}.  To effect this trade-off, we use linear combinations of the weighted and unweighted estimates, which we obtain by regularizing the survey weights in a DP manner. We use some of the  privacy budget to determine the degree of regularization and compute the statistics of interest.  We note that weight regularization is a familiar tool in survey contexts; for example, survey organizations routinely reduce anomalously large survey weights to reduce 
the variability of survey estimates \cite{gelman_struggles_2007,beaumont_new_2008,BNFP:SI15, prior-si2018}.  

\subsection{Contributions}
We summarize our contributions here:
\begin{enumerate}
    \item In Section \ref{sec:tri}, we analyze the three-way relationship between privacy loss, precision, and bias emerging from survey data. To do this, we introduce a regularization parameter $\lambda \in [0, 1]$ that linearly shrinks the survey weights to a constant when $\lambda = 1$. For any survey sample, there exists an ``optimal" value $\lambda^*$ which minimizes DP mean-squared error (for a fixed privacy loss) that depends on the sample size, survey measure range, possible weighting adjustments, and the difference between the unweighted and weighted mean estimates. We prove that $\lambda^* > 0$ (for any informative sampling design).  Similarly, we prove that for any finite privacy loss, there is a limit to the amount of bias that can be corrected by design-based weight adjustment without requiring DP noise that exceeds said correction. 
    \item In Section \ref{sec:algs}, we propose a two-step procedure to estimate survey-weighted population means using $\rho$-zero-concentrated differential privacy 
    \cite{bun_concentrated_2016}. First, we use the exponential mechanism to estimate $\lambda^*$; then, we use this output to shrink the survey weights and estimate the population mean using the Gaussian mechanism. We also provide different asymptotic and finite-sample approaches to quantifying errors due to sampling, weight shrinkage, and DP noise, allowing users to construct DP confidence intervals for the population mean estimates. 
       
    \item In Section \ref{sec:data_analysis}, we demonstrate our methodology on survey microdata from the Panel Study of Income Dynamics (PSID) \cite{psid}, a longitudinal survey containing family-level statistics on income sources and other sociodemographic information and oversampling from lower income sub-populations. We show how different survey outcome variables require different degrees of survey weight regularization, allowing analysts to more efficiently tailor DP privacy loss budgets when estimating multiple population means for different survey response variables. We also empirically evaluate uncertainty quantification properties, including confidence interval coverage.
\end{enumerate}

\subsection{Related Literature}

While there is an extensive literature on differentially private statistical analyses (see \cite{slavkovic_statistical_2023} for a review) and on methods for complex sample survey inference, there is little literature at their intersection. Many DP algorithms rely on the ``amplification by sub-sampling" property, wherein applying a DP algorithm on a simple random sample without replacement yields smaller privacy loss than the same algorithm applied to the entire population \cite{balle2018privacy}. However, for survey designs besides simple random sampling, this property may not hold \cite{bun_controlling_2022} nor does it always improve accuracy \cite{hu_accuracy_2022}. We consider design-based survey inference where the weights themselves contain all relevant sampling information and, therefore, must be protected with DP. We do not consider the release of auxiliary data used to construct design-based weights, instead isolating the privacy cost of incorporating survey sample design exclusively within the weights. 

The most direct line of work compared to ours uses methods that generate synthetic survey responses and weights \cite{hu_private_2022}. These methods can produce synthetic data that are interoperable with existing analyses and admit combining-rules-based approaches to inferences with synthetic data \cite{raghu:reiter:2007}.  Our approach differs in key ways. First, we directly use survey-weighted estimators as opposed to working through multiple synthetic datasets and inferences via combining rules. Second, we provide decision-making guidelines for whether certain kinds of weighting corrections can be sufficiently estimated using DP at a given finite sample size. 

\section{Background}
\label{sec:background}

In this section, we describe the finite population estimation setting and review the formal privacy notion that we use.

\subsection{Finite Population Estimation}
\label{sec:fpe}

Let $\mathcal{U}$ represent a finite population comprising $N$ records.  For $i=1, \dots, N$, let $y_i$ be a survey variable of interest, which we call the response variable. The data steward takes a random sample $\mathcal{S} \subset \mathcal{U}$.  For $i=1, \dots, N$, let $I_i=1$ when record $i$ is in $\mathcal{S}$ and let  $I_i=0$ otherwise.  Let $n=\sum_{i=1}^N I_i$ be the sample size.  For $i=1, \dots, N$, let $\pi_i = \p(I_i = 1)$ be the probability that record $i$ is randomly sampled.  This is determined from the sampling design features.  We also will utilize $\pi_{ij} = \p(I_i = 1, I_j=1)$, which is the joint probability that both records $i$ and $j$ are in $\mathcal{S}$.

The design weight for each record $i =1, \dots, N$ is defined as  $w_i =  1/\pi_i$.  Design weights often are adjusted, for example, for survey nonresponse \cite{NRBA-cr2022,NRBA-long2022}. For now, we work only with design weights.  We discuss incorporating adjusted weights in Section \ref{sec:discussion}.  Thus, the data available for estimation are $\mathcal{S} = \{(y_i, w_i): I_i=1, \, i=1, \dots, N\}$. We assume that $\{(y_i, w_i): I_i=0, \, i=1, \dots, N\}$ are not available for analysis.

We presume that both the response variable and the survey weights lie within the bounded intervals $[L_Y, U_Y]$ and $[L_W, U_W]$, respectively.  For response variables, this is naturally the case for categorical variables.  Bounds also can arise from the survey instrument, which may restrict the range of answers.  When neither of these is the case, there may be reasonable bounds from domain knowledge or variable definitions.  For design weights, at minimum $L_w=1$ in probability samples where every record in $\mathcal{U}$ has a nonzero chance of selection. The $U_W$ arises from the nature of the sampling design. 
In what follows, we presume that $U_W \geq N/n$, where $N/n$ is the weight that arises from the design that takes a simple random sample of $n$ records from the $N$ records in $\mathcal{U}$. This condition is typically satisfied in survey designs that are less efficient than simple random sampling, and it can be enforced when setting $U_W$. For convenience, we define $\Delta_W \triangleq U_W - L_W$ and, without loss of generality, assume $L_Y = 0$ and $1 \leq L_W \leq U_W$. In what follows, we will use $U_W$ as a conservative upper bound on $\Delta_W$. However, replacing $U_W$ with $\Delta_W$ in our downstream results will yield identical conclusions.

Our goal is to estimate the population mean $\theta \triangleq \sum_{i=1}^N y_i/N$.  With complex sample surveys, an unbiased estimator of $\theta$ is the survey-weighted mean \cite{horvitzthompson}
\begin{equation}
    \label{eq:wmean}
    \hat{\theta}(\bm{y}, \bm{w}) \triangleq \frac{1}{N} \sum_{i\in \mathcal{S}} y_i w_i.
\end{equation}
Here, each $y_i$ and $w_i$ are treated as fixed features of individual $i$.  The randomness in $\hat{\theta}$ derives from the probability sampling design, which is fully characterized by the survey weights. By treating each $y_i$ and $w_i$ as fixed constants, we make our analysis consistent with DP approaches that treat confidential data entries as constants from a fixed ``schema" of possible values.

We also can find an unbiased estimator of the variance of $\hat{\theta}$.  We have 
\begin{equation}
    \label{eq:ht_var}
    \widehat{\Var}_{\mathrm{HT}}(\hat{\theta})
    \triangleq \frac{1}{N^2} \left( \sum_{i \in \mathcal{S}} \frac{1 - \pi_i}{\pi_i^2} y_i^2 + \sum_{i \in \mathcal{S}} \sum_{j \neq i, j \in \mathcal{S}} \left( \frac{\pi_{ij} - \pi_i \pi_j}{\pi_i \pi_j} \frac{y_i y_j}{\pi_{ij}} \right) \right).
\end{equation}
When the second term in \eqref{eq:ht_var} is negative, as is usually the case in sampling-without-replacement designs, a conservative approximation uses only the first term in \eqref{eq:ht_var}, yielding the simplified estimator 
\begin{equation}
    \label{eq:ht_var_approx}
    \widehat{\Var}_{\mathrm{ApproxHT}}(\hat{\theta}) 
    \triangleq \frac{1}{N^2} \sum_{i \in \mathcal{S}} \frac{1 - \pi_i}{\pi_i^2} y_i^2.
\end{equation}
The estimator in \eqref{eq:ht_var_approx} is also the unbiased estimator of the variance of \ref{eq:wmean} for data collected by Poisson sampling, i.e., sample each record $i$ independently with its probability $\pi_i$, which by design has  
\begin{equation}
    \label{eq:poisson_approx}
    \pi_{ij} = \p(I_i = 1, I_j = 1) = \p(I_i = 1) \p(I_j = 1) = \pi_i \pi_j.
\end{equation}
Poisson sampling is typically not implemented in practice because it makes the sample size $n$ a random variable.  Agencies like to set $n$ in advance for budgetary and workload considerations.  Nonetheless, Poisson sampling estimators can be reasonable approximations for other unequal probability designs~\cite{sarn:wret:1992}.  

\subsection{Privacy Definitions and Methods}

In applications of formal privacy, one needs to define mathematically what one means by privacy, which we now do.  This requires a notion of what it means for two collected samples to be adjacent datasets, of which there are many possible options \cite{drechsler2024complexitiesdifferentialprivacysurvey}.

\begin{definition}[Adjacency]
    We say that two observed samples of size $n$ are adjacent if and only if they differ on the contributions of one observed record, i.e., 
    \begin{equation}
        \label{eq:adj}
        \{ (y_i, w_i) \}_{i=1}^n \sim_M \{ (y_i', w_i') \}_{i=1}^n \iff \#\{ i \in [n] \mid y_i \neq y_i' \text{ or } w_i \neq w_i' \} = 1.
    \end{equation}
 \end{definition}

Our analysis assumes that survey weights are fixed properties of individual records that do not change depending on which units appear in the realized sample. This helps align our analysis with standard DP analyses that treat observed confidential data (in this case, survey responses and weights) as constants, instead of random variables. We additionally assume that the population and sample sizes, $N$ and $n$, respectively, are public information. While this assumption reflects standard practice for publishing survey metadata, there may be confidentiality concerns if membership in the population under study is privacy-concerning, especially if the sample size $n$ is quite small.

 As a next step, we define the DP distance metrics that we use to construct the private finite population estimation methods. 

 \begin{definition}[$\rho$-zero-concentrated differential privacy \cite{bun_concentrated_2016}]
    Let $M$ be a randomized algorithm that releases statistics based on $\{ (y_i, w_i) \}_{i=1}^n$. We say that algorithm $M$ satisfies $\rho$-zero-concentrated differential privacy ($\rho$-zCDP) if, for all adjacent observed samples and all $\alpha \in (1, \infty)$, 
    \begin{equation}
        \label{eq:zcdp}
        d_{\alpha}\left(M(\bm{y}, \bm{w}) \ || \ M(\bm{y}', \bm{w}') \right) \leq \rho \alpha,
    \end{equation}
    where above, $d_\alpha(\cdot \ || \ \cdot)$ is the $\alpha$-R\'enyi divergence. 
\end{definition}

Broadly, DP provides numerous semantic privacy guarantees about the ability for adversaries to distinguish between two adjacent databases under various definitions of adjacency and data generating processes (see \cite{kifer_bayesian_2022} for an example discussion). Many of these guarantees rely on independence, i.e., assuming that data comes from an independent (but not necessarily identically) distributed process. In our context, \eqref{eq:poisson_approx} ensures that these semantic privacy guarantees hold because we assume that modifying one survey weight does not affect the others, both for our privacy analysis and our data generating assumption analysis.

We next briefly introduce some mechanisms which satisfy $\rho$-zCDP. Our method combines two common base algorithms used to satisfy DP: the Gaussian mechanism \cite{bun_concentrated_2016} and the exponential mechanism \cite{mcsherry_mechanism_2007}. Let $(\bm{y}', \bm{w}')$ be an arbitrary sample that is adjacent to $(\bm{y}, \bm{w})$.

\begin{definition}[Gaussian Mechanism \cite{bun_concentrated_2016}]
    Suppose the statistic 
    $$
    T: \{ [L_Y, U_Y] \times [L_W, U_W] \}^n \mapsto \mathbb{R}
    $$ 
    has sensitivity defined by
    \begin{equation}
        \Delta(T) \triangleq \sup_{(\bm{y}, \bm{w}) \sim_M (\bm{y}', \bm{w}')} \left| T(\bm{y}, \bm{w}) - T(\bm{y}', \bm{w}') \right|.
    \end{equation}
    Then, the mechanism $M$ defined as
    \begin{equation}
        M(\bm{y}, \bm{w}) = T(\bm{y}, \bm{w}) + \varepsilon, \qquad \varepsilon \sim N\left(0, \frac{\Delta(T)^2}{2\rho} \right),
    \end{equation}
    satisfies $\rho$-zCDP. 
\end{definition}
\medskip 

The exponential mechanism, presented in \cite{mcsherry_mechanism_2007}, 
satisfies $\epsilon$-DP \cite{dwork_calibrating_2006}, and we use the conversion lemma from \cite{bun_concentrated_2016} to modify its form below.

\begin{definition}[Exponential Mechanism \cite{mcsherry_mechanism_2007}]
    Suppose the goal is to minimize a real-valued loss function $\ell$ over output space $\mathcal{Z}$,
    $$
    \ell: \mathcal{Z} \times \{ [L_Y, U_Y] \times [L_W, U_W] \}^n   \mapsto [0, \infty).
    $$
    Similarly, we define a functional analogue of the sensitivity given by
    \begin{equation}
        \Delta(\ell) \triangleq \sup_{(\bm{y}, \bm{w}) \sim_M (\bm{y}', \bm{w}'), \, z \in \mathcal{Z}}  |\ell(z; \bm{y}, \bm{w}) - \ell(z; \bm{y}', \bm{w}')|. 
    \end{equation}
    Then, releasing one sample from the distribution over $\mathcal{Z}$ with a density given by
    \begin{equation}
        f(z) \propto \exp\left(- \frac{\sqrt{2\rho} }{2 \Delta(\ell)} \ell(z; \bm{y}, \bm{w}) \right)
    \end{equation}
    satisfies $\rho$-zCDP.
\end{definition}

\section{DP Finite Population Estimation}
\label{sec:methods}

In this section, we propose the DP algorithm for releasing survey weighted point and interval estimates. We begin by characterizing the three-way trade-off in bias, precision, and privacy inherent in DP estimation with finite population quantities.

\subsection{Bias-Precision-Privacy Trade-offs}
\label{sec:tri}

As a motivation for our methods, consider applying the Gaussian mechanism to a survey-weighted estimate of the form in \eqref{eq:wmean}.  Let $\hat{\theta}(\bm{y}, \bm{w})$ be the value of \eqref{eq:wmean} computed using some particular $\mathcal{S} = (\bm{y}, \bm{w})$. In this setting, the sensitivity of the weighted estimator is given by
\begin{equation}
    \Delta(\hat{\theta}) = \sup_{(\bm{y}, \bm{w}) \sim_M (\bm{y}', \bm{w}')} \left| \hat{\theta}(\bm{y}, \bm{w}) - \hat{\theta}(\bm{y}', \bm{w}') \right| = \frac{U_W U_Y}{N}.
\end{equation}
To satisfy $\rho$-zCDP, we could release
\begin{equation}
    \hat{\theta}^{(\rho )}(\bm{y}, \bm{w})  \triangleq \hat{\theta}(\bm{y}, \bm{w}) + \varepsilon, \qquad \varepsilon \sim N\left(0, \frac{\Delta(\hat{\theta})^2}{2\rho} \right).
\end{equation}

This approach requires Gaussian noise with variance that grows with $\Delta_W^2$, which could result in impractical results when $\Delta_W$ is large.
Such issues are especially pronounced for surveys, where sample sizes are typically much smaller than those used to evaluate DP algorithms \cite{drechsler_differential_2023}. 

 As a start towards an alternative, now consider the unweighted estimate of the mean, which we write as $\hat{\theta}_0$. Note that we can obtain $\hat{\theta}_0$ from \eqref{eq:wmean} by replacing all the design weights for records in $\mathcal{S}$ with $N/n$. We can express $\hat{\theta}$ as  
\begin{equation}
    \hat{\theta} = \hat{\theta}_0 + \mathrm{Sign}(\hat{\theta} - \hat{\theta}_0) |\hat{\theta} - \hat{\theta}_0|.\label{eq:thetatotheta0}
\end{equation}
The first term in \eqref{eq:thetatotheta0} is the standard, low-sensitivity unweighted mean for which classical DP release mechanisms offer optimal utility guarantees \cite{awan2023canonical}. The second term in  \eqref{eq:thetatotheta0} comprises the product of $\mathrm{Sign}(\hat{\theta} - \hat{\theta}_0)$, which we call the weighting discrepancy sign, and $|\hat{\theta} - \hat{\theta}_0|$, which we call the 
absolute weighting discrepancy (AWD).

The AWD has high sensitivity and makes DP survey estimation difficult in practice.  
However, the actual value of the AWD can be close to zero in some scenarios, particularly when 
the survey response variable and survey weights are
nearly uncorrelated \cite{little-weighting-SM2005}. 
Thus, in cases with sufficiently low AWD, it may be preferable to use $\hat{\theta}_0$ and disregard the inflation in sensitivity from accounting for the survey weights. 

This thought experiment suggests it may be fruitful to define estimators that use linear combinations of $\hat{\theta}$ and $\hat{\theta}_0$ for DP release.  To do so, we introduce a regularization parameter $\lambda \in [0, 1]$ of the form
\begin{eqnarray}
 \nonumber   \hat{\theta}_\lambda &\triangleq& 
    \lambda \hat{\theta}_0 + (1 - \lambda) \hat{\theta} 
    = \hat{\theta}_0 + (1 - \lambda) \ \mathrm{Sign}(\hat{\theta} - \hat{\theta}_0) |\hat{\theta} - \hat{\theta}_0|.
\end{eqnarray}
Essentially, this estimator reduces reliance on the highly sensitive AWD in favor of the less insensitive unweighted mean.

We can instantiate $\hat{\theta}_{\lambda}$ by defining new weights for use in \eqref{eq:wmean}.  These new weights are created via the function $G_\lambda$, where 
\begin{equation}
    \label{eq:lambda} 
    G_\lambda(\bm{w}) \triangleq (1 - \lambda) \bm{w} + \frac{\lambda N}{n} \mathbbm{1}_{n}, \qquad \hat{\theta}_{\lambda}(\bm{y}, \bm{w}) = \hat{\theta}(\bm{y}, G_\lambda(\bm{w})).
\end{equation}
This function can be interpreted as shrinking the design survey weights towards $N/n$, whereby $\lambda = 0$ corresponds to the original weights and $\lambda = 1$ corresponds to $N/n$. For notational convenience, we will sometimes apply $G_\lambda(\cdot)$ to scalar values instead of vector values, such as $G_\lambda(U_W)$.

The sensitivity of $\hat{\theta}_{\lambda}$ is given by
\begin{align}
    \Delta(\hat{\theta}_\lambda) &= \sup_{(\bm{y}, \bm{w}) \sim_M (\bm{y}', \bm{w}')} \left| \hat{\theta}_\lambda(\bm{y}, \bm{w}) - \hat{\theta}_\lambda(\bm{y}', \bm{w}') \right| = \frac{G_\lambda(U_W) U_Y}{N}.
\end{align}
The reduced sensitivity admits the new $\rho$-zCDP estimator,
\begin{equation}
    \hat{\theta}_\lambda^{(\rho)}(\bm{y}, \bm{w})  \triangleq \hat{\theta}(\bm{y}, G_\lambda(\bm{w})) + \varepsilon, \qquad \varepsilon \sim N\left(0, \frac{\Delta(\hat{\theta}_\lambda)^2}{2\rho} \right).
\end{equation}

This reduction in sensitivity comes at a cost determined by the difference between $\hat{\theta}(\bm{y}, G_\lambda(\bm{w}))$ and $\hat{\theta}(\bm{y}, \bm{w})$. We give this quantity a name, \emph{mechanism bias}, to quantify bias induced by the DP mechanism, defined as
\begin{equation}
    D(\lambda) \triangleq \Ex_{\varepsilon}[ \hat{\theta}_\lambda^{(\rho )}(\bm{y}, \bm{w})] - \hat{\theta}(\bm{y}, \bm{w}).
\end{equation}
This quantity measures the expected difference between the DP estimator and the unbiased non-DP estimator. For our proposed regularization strategy, the mechanism bias is linear in $\lambda$, i.e.,
\begin{align}
 \nonumber     D(\lambda) &= 
    \hat{\theta}(\bm{y}, G_{\lambda}(\bm{w})) - \hat{\theta}(\bm{y}, \bm{w}) =
     \frac{1}{N} \left[ \sum_{i=1}^n \left( (1- \lambda) w_i + \frac{\lambda N}{n} - w_i \right) y_i \right] \\
    &= \lambda \left[\frac{1}{n} \sum_{i=1}^n y_i - \frac{1}{N} \sum_{i=1}^n y_i w_i \right] = \lambda (\hat{\theta}_0 - \hat{\theta}).
    \label{eq:lambda_linear_bias} 
\end{align}
Therefore, conditional on $\mathcal{S}$, the mean-squared error (MSE) introduced by DP takes the form
\begin{align}
 \nonumber     \ell(\lambda; \bm{y}, \bm{w}) &\triangleq 
    \Ex_{\varepsilon}\left[ (\hat{\theta}_{\lambda}^{(\rho)} - \hat{\theta}_{\lambda} + \hat{\theta}_{\lambda} -\hat{\theta})^2 \right] \\
  \nonumber    &= \Ex_{\varepsilon}\left[ (\hat{\theta}_{\lambda}^{(\rho)} - \hat{\theta}_{\lambda})^2 \right] + (\hat{\theta}_{\lambda} -\hat{\theta})^2 + 2 \Ex_{\varepsilon}\left[ \hat{\theta}_{\lambda}^{(\rho)} - \hat{\theta}_{\lambda} \right] \left( \hat{\theta}_{\lambda} -\hat{\theta}\right) \\
    &= \frac{\Delta(\hat{\theta}_{\lambda})^2}{2\rho} + D(\lambda)^2.\label{eq:loss_tri}
\end{align}
The result in  \eqref{eq:loss_tri} characterizes a three-way trade-off in bias, precision, and privacy for DP survey estimation. As we reduce the mechanism bias of our survey estimates, we require more additive noise to satisfy $\rho$-zCDP. Moreover, this effect becomes more extreme as $\rho$ gets smaller. 

For a fixed value of $\rho$, we can 
minimize \eqref{eq:loss_tri} as a function of $\lambda$ as shown in Lemma \ref{lem:tri}.
\begin{lemma}
    \label{lem:tri}
    Consider the loss function $\ell(\lambda; \bm{y}, \bm{w})$ in \eqref{eq:loss_tri}. As a function of $\lambda$, the $\ell(\lambda; \bm{y}, \bm{w})$ 
        is minimized by 
        \begin{equation}
            \label{eq:lambda_minn}
            \lambda^* \triangleq \min\left\{1, \frac{\frac{U_W}{\rho} \left( \frac{U_Y}{N} \right)^2 \left(U_W - \frac{N}{n} \right)}{\left[ \frac{1}{\rho} \left( \frac{U_Y}{N} \right)^2 \left(U_W - \frac{N}{n} \right)^2 + 2 (\hat{\theta}_0 - \hat{\theta})^2 \right]} \right\}.
        \end{equation}
\end{lemma}
 Note that $\lambda^* > 0$ when $U_W > N/n$.  Thus, for scenarios where $U_W > N/n$, it is never optimal to use survey weights as-is; some degree of regularization, i.e., $\lambda^* > 0$, provides lower error. Also note that $\lambda^* < 1$ when 
        \begin{equation}
            \label{eq:lambda_nontriv_condition2}
            \left| \hat{\theta}_0 - \hat{\theta} \right| > \sqrt{ \frac{U_Y^2}{2\rho N n} \left(U_W - \frac{N}{n} \right) }.
        \end{equation}
        Equivalently, $\lambda^* < 1$ when 
        \begin{equation}
            \label{eq:lambda_nontriv_condition3}
            \rho > \frac{U_Y^2}{2(\hat{\theta}_0 - \hat{\theta} )^2 N n} \left(U_W - \frac{N}{n} \right).
        \end{equation}
Thus, if the effect of the mechanism bias is not sufficiently large, or if $\rho$ is sufficiently small relative to $U_Y^2$, then the optimal strategy to minimize $\ell(\lambda \cdot; \bm{y}, \bm{w})$ is to ignore the survey weights entirely, i.e., set $\lambda^* = 1$.

\begin{figure}[t]
    \centering
    \includegraphics[width=.99\textwidth]{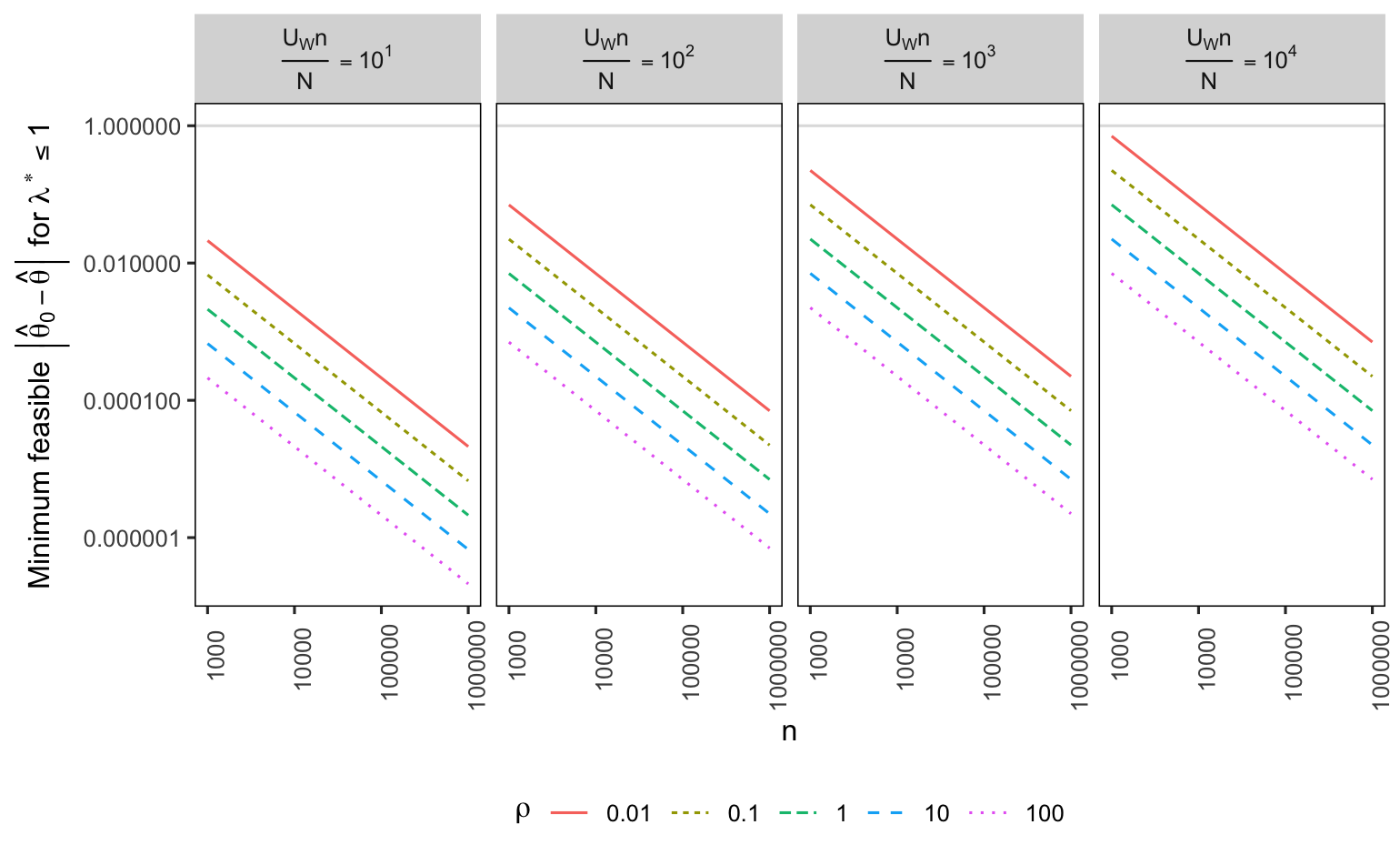}
    \caption{Minimum feasible values for $\left| \hat{\theta}_0 - \hat{\theta} \right|$ as a function of sample size $n$, weight ratio $U_W n / N$, and privacy loss $\rho$, where $\theta$ is the mean of a binary variable in a population of size $N=10^8$. }
    \label{fig:min_feas_shrinkage_example}
\end{figure}

We illustrate this trade-off using hypothetical inputs to Lemma \ref{lem:tri}. Suppose we are interested in estimating a proportion with a binary
response $y_i \in \{ 0, 1 \}$ 
from a population of size $N = 10^8$. We consider varying the sample size $n$, the privacy loss budget for estimating $\rho$, and the weight ratio $U_W/(N/n)$.
In each case, we calculate the minimum value for the AWD, $\left| \hat{\theta}_0 - \hat{\theta}\right|$, where equality is achieved in  \eqref{eq:lambda_nontriv_condition2}. This value represents the minimum difference in the population proportion with or without using survey weights necessary to consider accounting for the weighting process in a DP estimate. Figure \ref{fig:min_feas_shrinkage_example} displays the results. As expected, DP estimators for the population mean can better incorporate weighting information as the sample size increases, the weight ratio decreases, and  $\rho$ increases.  However, when these parameters move in opposite directions, it becomes harder to justify incorporating survey weights into the analysis. For example, in the extreme case where the weight ratio is $10^4$ (meaning one respondent can hypothetically contribute up to $10^4$ more to a weighted mean) and $\rho = 1$, we require at least a 10\% difference between the weighted and unweighted statistics for a sample of $n=10^3$ respondents to justify incorporating survey weights. Results like these can thus help determine the kinds of survey-weighting corrections that are feasible or infeasible to consider with DP. 

\subsection{DP Algorithm for Survey Weight Regularization}
\label{sec:algs}

The optimal degree of regularization $\lambda^*$ depends on the confidential data through the AWD. Therefore, we propose the following two-step approach to estimating $\hat{\theta}$ using $(\rho_1 + \rho_2)$-zCDP, written out in Algorithm \ref{alg:priv_reg_est}. First, we estimate $\lambda^*$ while satisfying $\rho_1$-zCDP via the exponential mechanism. We sample a value $\hat{\lambda}^{(\rho_1)}$ from the density
    \begin{equation}
        \label{eq:lambda_expmech_implement}
        f(\lambda) \propto \ind{\lambda \in [0, 1]} \exp\left(-\frac{\sqrt{2\rho_1}}{2\Delta(\ell)} \ell(\lambda; \bm{y}, \bm{w}) \right),
    \end{equation}
    where  $\Delta(\ell) = (\Delta(\hat{\theta}) - \Delta(\hat{\theta}_0))^2$ as shown in the Appendix \ref{apx:proofs}.
Second, 
we sample $\hat{\theta}_{\hat{\lambda}^{(\rho_1)}}^{(\rho_2)}$ according to the Gaussian mechanism using weights computed from \eqref{eq:lambda} with the sampled $\hat{\lambda}^{(\rho_1)}$. In this place, we leverage Lemma \ref{lem:tri} and replace its $\rho$ with $\rho_2$ in our proposed algorithm. 
In this way, we still satisfy $(\rho_1 + \rho_2)$-zCDP per Theorem \ref{th1}, proved in the Appendix, while maintaining a high probability of reducing the 
DP noise added to the estimate. 

\begin{theorem}\label{th1}
    Algorithm \ref{alg:priv_reg_est} satisfies $(\rho_1 + \rho_2)$-zCDP.
\end{theorem}

\begin{algorithm}[t]
\caption{DP regularized survey-weighted population estimate}
\label{alg:priv_reg_est}
\begin{algorithmic}
    \REQUIRE $\rho_1, \rho_2 \in (0, \infty)$, $\{ (y_i, w_i) \}_{i=1}^n \in \left\{[L_Y, U_Y] \times [L_W, U_W] \right\}^n$, $N \in \mathbb{N}$. 
    \STATE Sample $\hat{\lambda}^{(\rho_1)}$ from the density $f(\lambda)$ where
    \begin{equation}
        \label{eq:alg1}
        f(\lambda) \propto \ind{\lambda \in [0, 1]} \exp\left(-\frac{\sqrt{2\rho_1}}{2\Delta(\ell)} \left( \frac{\Delta(\hat{\theta}_{\lambda})^2}{2\rho_2} + D(\lambda)^2 \right) \right).
    \end{equation}
    Sample $\hat{\theta}_{\hat{\lambda}^{(\rho_1)}}^{(\rho_2)}$ where
    \begin{equation}
        \label{eq:alg2}
        \hat{\theta}_{\hat{\lambda}^{(\rho_1)}}^{(\rho_2)} \sim N\left( \hat{\theta}_{\hat{\lambda}^{(\rho_1)}}(\bm{y}, \bm{w}), \frac{1}{2\rho_2} \left( \frac{1}{N} \left[ G_{\hat{\lambda}^{(\rho_1)}}(U_W) U_Y \right] \right)^2 \right).
    \end{equation}
    \RETURN $(\hat{\lambda}^{(\rho_1)}, \hat{\theta}_{\hat{\lambda}^{(\rho_1)}}^{(\rho_2)})$.
\end{algorithmic}
\end{algorithm}

We can characterize the errors introduced by Algorithm \ref{alg:priv_reg_est} by quantifying the concentration of  $\hat{\theta}_{\hat{\lambda}^{(\rho_1)}}^{(\rho_2)}$ about $\hat{\theta}$. The contribution of the Gaussian mechanism noise error is trivial to quantify. For the contribution from the mechanism bias, we take the following strategy.  Using  $\hat{\lambda}^{(\rho_1)}$, we approximate the smallest value that $\hat{\lambda}^*$ might be while still giving $\hat{\lambda}^{(\rho_1)}$ sufficiently high probability. Since $\lambda^*$ decreases as AWD increases, finding a plausible lower bound for  $\hat{\lambda}^*$ 
allows us to find a plausible upper bound for AWD.

This logic yields the result in Theorem \ref{thm:dpreg_util}.  Let  $\sigma^{2*}(\bm{y}, \bm{w})$ be the variance  of $\hat{\lambda}^{(\rho_1)}$ obtained from $f(\lambda)$ in \eqref{eq:lambda_expmech_implement}.

\begin{theorem}
    \label{thm:dpreg_util}
    Let 
        $(\hat{\lambda}^{(\rho_1)},  \hat{\theta}_{\hat{\lambda}^{(\rho_1)}}^{(\rho_2)})$
    be the output of Algorithm \ref{alg:priv_reg_est}. Then, we have
    \begin{equation}
        \p\left( |\hat{\theta}_{\hat{\lambda}^{(\rho_1)}}^{(\rho_2)} - \hat{\theta}| \leq C^* \right) \leq 1 - \alpha
    \end{equation}
    where 
    \begin{equation}
        C^* = \frac{D^-\left(\hat{\lambda}^{(\rho_1)} + z_{\alpha/4} \sqrt{\sup_{\bm{y}, \bm{w}} \left[ \sigma^{2*}(\bm{y}, \bm{w}) \right]}\right)}{\hat{\lambda}^{(\rho_1)}} + z_{\alpha/4} \frac{G_{\hat{\lambda}^{(\rho_1)}}(U_W) U_Y}{N\sqrt{2\rho_2}} 
    \end{equation}
    \begin{equation}
    \label{eq:weight_bias_mono3} 
    D^{-}(\lambda) \triangleq \sqrt{ \frac{1}{2} \left[ \frac{\frac{U_W}{\rho_2} \left( \frac{U_Y}{N} \right)^2 \left(U_W - \frac{N}{n} \right)}{\lambda} - \frac{1}{\rho_2} \left( \frac{U_Y}{N} \right)^2 \left(U_W - \frac{N}{n} \right)^2 \right] }. 
\end{equation}
\end{theorem}
\begin{proof} 

To begin, we note that 
$\hat{\lambda}^{(\rho_1)}$  follows a truncated normal distribution on $[0, 1]$ with the variance parameter
\begin{equation}
    \sigma^{2*}(\bm{y}, \bm{w}) \triangleq \left( \frac{\sqrt{2\rho_1}}{\Delta(\ell)} 
    \left[ \frac{1}{2\rho_2} \left(\frac{\Delta_Y}{N} \right)^2 \left(\frac{N}{n} - U_W \right)^2 + (\hat{\theta}_0 - \hat{\theta})^2\right] \right)^{-1}. 
\end{equation}
Independent of any confidential data, we have
\begin{equation}
    \sup_{\bm{y}, \bm{w}} \left[ \sigma^{2*}(\bm{y}, \bm{w}) \right] \triangleq \left( \frac{\sqrt{2\rho_1}}{\Delta(\ell)} 
    \left[ \frac{1}{2\rho_2} \left(\frac{\Delta_Y}{N} \right)^2 \left(\frac{N}{n} - U_W \right)^2 \right] \right)^{-1}.
\end{equation}

Thus, using these properties of $f(\lambda)$, we have (over the randomness in $\hat{\lambda}^{(\rho_1)}$), 
\begin{equation}
    \p\left(\lambda^* \geq \hat{\lambda}^{(\rho_1)} + z_{\alpha/2} \sqrt{\sup_{\bm{y}, \bm{w}} \left[ \sigma^{2*}(\bm{y}, \bm{w}) \right]} \right) \geq 1 - \alpha.\label{pforlambda}
\end{equation}
When $\lambda^* \in (0, 1)$, we can rearrange \eqref{eq:lambda_minn} to show that 
\begin{equation}
    \label{eq:weight_bias_mono2} |\hat{\theta}_0 - \hat{\theta}| = \sqrt{ \frac{1}{2} \left[ \frac{\frac{U_W}{\rho_2} \left( \frac{U_Y}{N} \right)^2 \left(U_W - \frac{N}{n} \right)}{\lambda^*} - \frac{1}{\rho_2} \left( \frac{U_Y}{N} \right)^2 \left(U_W - \frac{N}{n} \right)^2 \right] } \triangleq D^-(\lambda^*).
\end{equation} 
Using \eqref{eq:weight_bias_mono2}, we can rewrite \eqref{pforlambda} as  
\begin{equation}
    \p\left(|\hat{\theta}_{0} - \hat{\theta}| \leq D^-\left(\hat{\lambda}^{(\rho_1)} + z_{\alpha/2} \sqrt{\sup_{\bm{y}, \bm{w}} \left[ \sigma^{2*}(\bm{y}, \bm{w}) \right]}\right) \right) \geq 1 - \alpha,
\end{equation}
and therefore,
\begin{equation}
    \p\left(|\hat{\theta}_{\hat{\lambda}^{(\rho_1)}} - \hat{\theta}| \leq \frac{D^-\left(\hat{\lambda}^{(\rho_1)} + z_{\alpha/2} \sqrt{\sup_{\bm{y}, \bm{w}} \left[ \sigma^{2*}(\bm{y}, \bm{w}) \right]}\right)}{\hat{\lambda}^{(\rho_1)}} \right) \geq 1 - \alpha.
\end{equation}
\end{proof}

We make a few comments about Theorem \ref{thm:dpreg_util}. First, we used a union bound to establish our concentration around the two parameters. Depending on prior beliefs about the relative magnitude of errors due to mechanism bias or Gaussian noise addition, one could place more weight on either component in the concentration inequality. Second,  we can use the estimated regularization parameter as a noisy ``plug-in`` proxy for estimating AWD with $\rho_1$-zCDP (with the caveat that this plug-in estimator is high-sensitivity, like our original estimand). To do this, suppose $\hat{\theta}_0 > \hat{\theta}$; then
\begin{equation}
    \label{eq:plugin_bias}
    |\hat{\theta}_0-\hat{\theta}| \approx D^-(\hat{\lambda}^{(\rho_1)}) \implies \Ex\left[ \hat{\theta}_{\hat{\lambda}^{(\rho_1)}}^{(\rho_2)} + \hat{\lambda}^{(\rho_1)} D^-(\hat{\lambda}^{(\rho_1)}) \right] \approx \hat{\theta}.
\end{equation}
Finally, Theorem \ref{thm:dpreg_util} does not directly consider the weighting bias sign, $\mathrm{Sign}(\hat{\theta} - \hat{\theta}_0)$, as only the AWD appears in the proof. If this sign is unknown {\em a priori}, one can modestly estimate it using an instantiation of the exponential mechanism. However, in practice, it may be reasonable to treat the sign as public information. For example, our analysis of the PSID considers data where it is known {\em a priori} that the survey oversamples families from lower incomes, so that unweighted estimates of mean income are almost certainly biased low. 

\subsection{Accounting for Sampling Error}

Under the classical, mild asymptotic conditions in \cite{berger1998rate}, as sample size $n$ increases relative to population size $N$, we have $\hat{\theta} \to_D N\left(\theta, \Var_{\mathrm{HT}}(\hat{\theta}) \right)$. This supports asymptotically consistent confidence intervals using consistent  
estimators for the mean and variance. To ease the computation of sensitivity for DP, as discussed in Section \ref{sec:fpe}, absent privacy concerns we use $\widehat{\Var}_{\mathrm{ApproxHT}}(\hat{\theta})$ from \eqref{eq:ht_var} to approximate $\Var_{\mathrm{HT}}(\hat{\theta})$. 

However, we apply zCDP to  $\widehat{\Var}_{\mathrm{ApproxHT}}(\hat{\theta})$, as it depends on the confidential data.
We do so by spending an additional $\rho_3$ of privacy loss, resulting in the variance estimate $\hat{V}^{(\rho_3)}$ which we use to construct end-to-end DP 
intervals. 

Algorithm \ref{alg:priv_reg_ci} summarizes the procedure. We do not use the regularized weights in the variance. This would necessarily
underestimate the sampling variance, since contributions to the sample variance estimate increase as $w_i^2$ increases. As a result, we provide a conservative estimate of the sampling variance.

Additionally, we introduce a new parameter $\alpha_v \in (0, 1)$ to calculate a $(1 - \alpha_v)\%$ upper bound on the confidential $\widehat{\Var}_{\mathrm{ApproxHT}}(\hat{\theta})$. This parameter is used in Algorithm \ref{alg:priv_reg_ci} to conservatively estimate the sampling variance in \eqref{eq:dpci}, account for DP noise in $\hat{V}^{(\rho_3)}$ when constructing the intervals. Note that for small values of $\rho_3$, there is a small non-zero probability that this quantity will drop below 0. We truncate $\hat{V}^{(\rho_3)}$ from below at 0 to prevent the square root in  \eqref{eq:dpci} from being negative, while still accounting for DP noise in the variance estimate.

\begin{algorithm}[t]
\caption{DP regularized survey-weighted population confidence interval}
\label{alg:priv_reg_ci}
\begin{algorithmic}
    \REQUIRE $\rho_1, \rho_2, \rho_3 \in (0, \infty)$, $\{ (y_i, w_i) \}_{i=1}^n \in \left\{[L_Y, U_Y] \times [L_W, U_W] \right\}^n$, $N \in \mathbb{N}$, $\alpha \in (0, 1)$, $\alpha_v \in (0, 1)$. 
    \STATE Sample $\hat{\lambda}^{(\rho_1)}$ and $\hat{\theta}_{\hat{\lambda}^{(\rho_1)}}^{(\rho_2)}$ according to Algorithm \ref{alg:priv_reg_est}. \\
    \STATE Sample
    \begin{equation}
        \label{eq:alg3}
        \hat{V}^{(\rho_3)} \sim N\left(  \widehat{\Var}_{\mathrm{ApproxHT}}(\hat{\theta}), \ \frac{\Delta(\hat{V})^2}{2\rho_3} \right), \qquad \Delta(\hat{V}) = \Delta(\hat{\theta})^2.
    \end{equation}
    \RETURN 
    \begin{equation}
        \label{eq:dpci}
        \hat{\theta}_{\hat{\lambda}^{(\rho_1)}}^{(\rho_2)} \pm z_{\alpha/2} \sqrt{\frac{\Delta(\hat{\theta}_{\hat{\lambda}^{(\rho_1)}})^2}{2\rho_2} + \max\{ \hat{V}^{(\rho_3)}, 0 \} + z_{\alpha_v/2} \sqrt{\frac{\Delta(\hat{V})^2}{2\rho_3}}}.
    \end{equation}
\end{algorithmic}
\end{algorithm}

\begin{theorem}
    Let $\hat{\theta}^*$ be an arbitrary new population mean estimate, drawn from the same sampling scheme. Under the regularity conditions in \cite{berger1998rate} as $N,n \to \infty$ 
    \begin{equation}
        \p\left(\hat{\theta}^* \in \left[\hat{\theta}_{\hat{\lambda}^{(\rho_1)}}^{(\rho_2)} \pm z_{\alpha/2} \sqrt{\frac{\Delta(\hat{\lambda}^{(\rho_1)})^2}{2\rho_2} + \hat{V}^{(\rho_3)}  + z_{\alpha_v/2} \sqrt{\frac{\Delta(\hat{V})^2}{2\rho_3}}} \right] \right) \to_P 1 - \alpha.
    \end{equation}
\end{theorem}
\begin{proof}
    As $N,n \to \infty$ under the conditions in \cite{berger1998rate}, the mechanism bias converges to 0, i.e.,
    \begin{equation}
        \hat{\lambda}^{(\rho_1)} \to_P 0 \qquad \implies \qquad \hat{\theta}_{\lambda_{\rho_1}} \to_P \hat{\theta}.
    \end{equation}
    Similarly, the sensitivities of $\hat{\theta}_{\hat{\lambda}^{(\rho_1)}}^{(\rho_2)}$ and $\hat{V}^{(\rho_3)}$ decrease to 0 as $N,n \to \infty$ under the conditions in \cite{berger1998rate}. Combining all the above, we have that the interval in  \eqref{eq:dpci} converges to the classical non-DP normal approximation interval, yielding the result.
\end{proof}

Note that the asymptotic consistency of the DP confidence interval does not require adjustments for uncertainty in $\hat{\lambda}^{(\rho_1)}$. Instead, we rely on the estimation of $\hat{\lambda}^{(\rho_1)}$ to determine the degree to which AWD might affect our inferences. 

Finally, for completeness, we state the privacy guarantee of Algorithm \ref{alg:priv_reg_ci} in \eqref{th:DP2}.
\begin{theorem}
    Algorithm \ref{alg:priv_reg_ci} satisfies $(\rho_1 + \rho_2 + \rho_3)$-zCDP.\label{th:DP2}
\end{theorem}

\section{Empirical Analyses}
\label{sec:data_analysis}


To demonstrate the methodology, we analyze family-level data from the PSID published in the year 2019. 
The PSID facilitates research on income, wealth, employment, health, family and child development, and other sociodemographic and economic topics, with an oversample of lower-income families.
We use data comprising $n=9420$ families sampled from a population of $N \approx 1.29 \times 10^8$ families. 
For the purposes of this evaluation, we treat the provided survey weights as design-based (in reality they are adjusted for nonresponse and coverage errors; see \cite{psid} for details). Under this weighting scheme, we are setting $U_W = 6 \times 10^4$ as a conservative upper bound on our survey weights. We seek to estimate population means of the variables in Table \ref{tab:psid_vars} using $\rho$-zCDP. We also add one simulated variable, \texttt{bern}, which contains iid Bernoulli draws to simulate a random survey response that is theoretically independent of the survey weights.

\begin{table}[t]
    \centering
    \begin{tabular}{c|c|c|c}
        \hline
        Variable & Description & $U_Y$ & $\hat{\theta}_0 - \hat{\theta}$ \\
        \hline 
        \texttt{inc3} & Cube-root-transformed family income & 150 & -.67 \\
        \hline 
        \texttt{pov} & 1 if family income below poverty line, else 0 & 1 & .022 \\
        \hline 
        \texttt{nf} & Number of family members & 20 & .27 \\
        \hline
        \texttt{bern} & iid Bernoulli(.5) random draws & 1 & .004 \\
        \hline 
    \end{tabular}
    \caption{Selected PSID variables and the simulated Bernoulli variable.}
    \label{tab:psid_vars}
\end{table}

 Figure \ref{fig:psid_public_values} displays the relationship between survey weights and the transformed income variable, \texttt{inc3}. The 
survey weights are weakly correlated with family income, with a Spearman's rank correlation of approximately .14. Similarly, Figure \ref{fig:psid_pov_weights} displays the distribution of survey weights for families below and above the poverty line. Here, we see a visible difference in distributions. Differences in $\hat{\theta}_0-\hat{\theta}$ are shown in Table \ref{tab:psid_vars}.  By ignoring survey weights, we would underestimate the national average family income and overestimate the national poverty rate.

\begin{figure}[t]
    \centering
    \includegraphics[width=.8\textwidth]{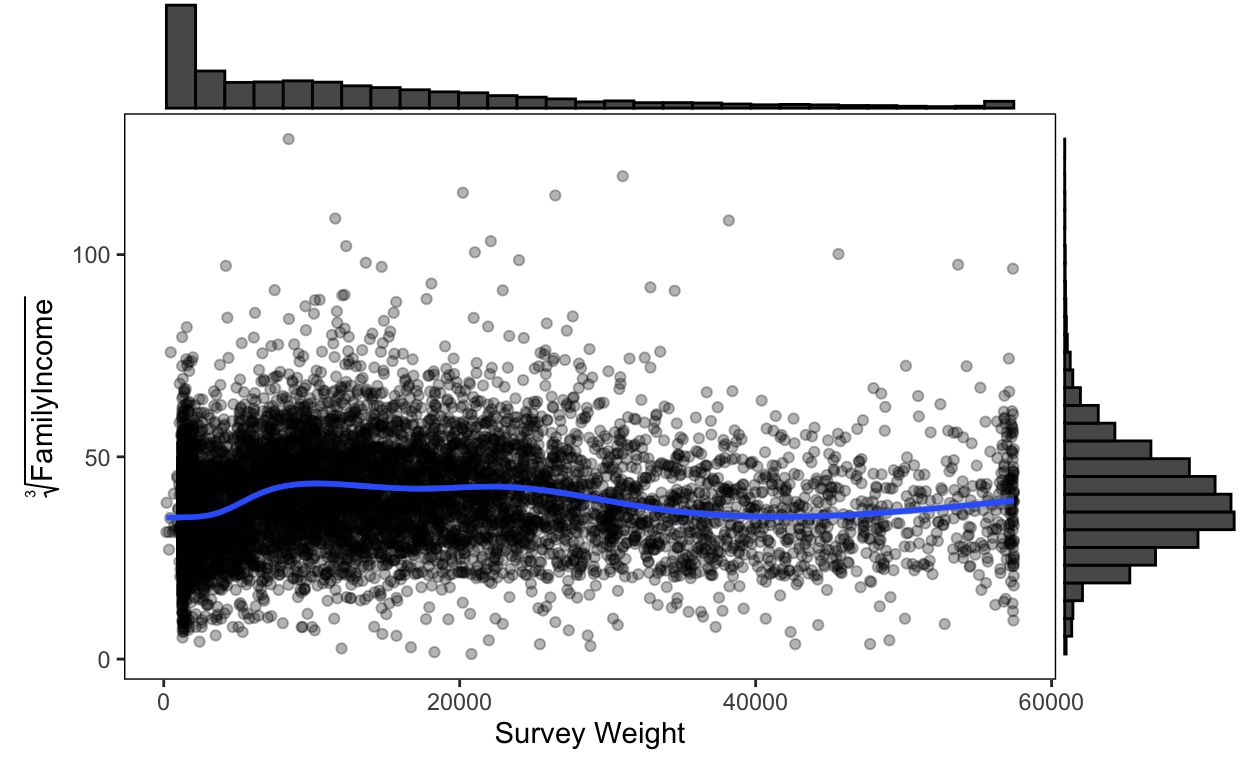}
    \caption{Plot of survey weights (x-axis) and $\texttt{inc3}$ (y-axis), with univariate histograms on the margins and a spline estimate of the central tendency in blue.}
    \label{fig:psid_public_values}
\end{figure}

\begin{figure}[h]
    \centering
    \includegraphics[width=.8\textwidth]{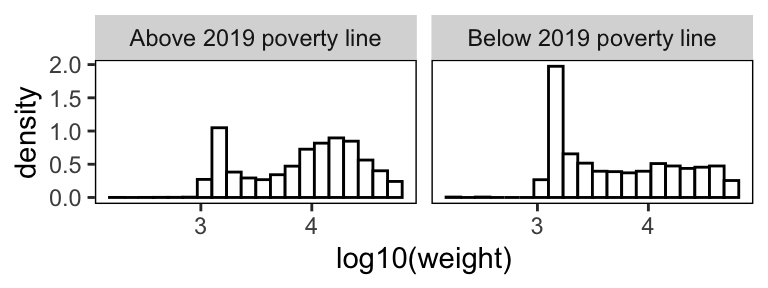}
    \caption{Histograms of logarithm transformed survey weights for respondents above and below the 2019 poverty lines (left and right, respectively).}
    \label{fig:psid_pov_weights}
\end{figure}

\subsection{Analysis of Privacy-Precision-Bias Trade-Offs}

In our first set of simulations, we use the results from Lemma \ref{lem:tri} to show how spending more privacy loss enables more fine-grained AWD corrections. Figure \ref{fig:psid_feas} is similar to Figure \ref{fig:min_feas_shrinkage_example}  but plots realized values from the two income-related PSID response variables. For demonstration purposes, we 
vary the sample size $n$ (represented by the different colored lines). The $x$-axis represents the privacy loss budget spent on estimating the population mean, and the $y$-axis refers to the smallest possible weighting bias for which $\lambda^* < 1$, i.e., for which we still benefit from accounting for survey weighting in DP inference. As expected, smaller AWD values can be accommodated with larger sample sizes and privacy loss budgets. The horizontal dashed lines refer to the realized AWDs for the two variables of interest (in a pure DP analyses, these would be confidential). Values above these lines refer to realized biases that would admit informative bias corrections at the allowed privacy loss budget level \emph{a priori}. 

\begin{figure}
    \centering
    \includegraphics[width=.99\textwidth]{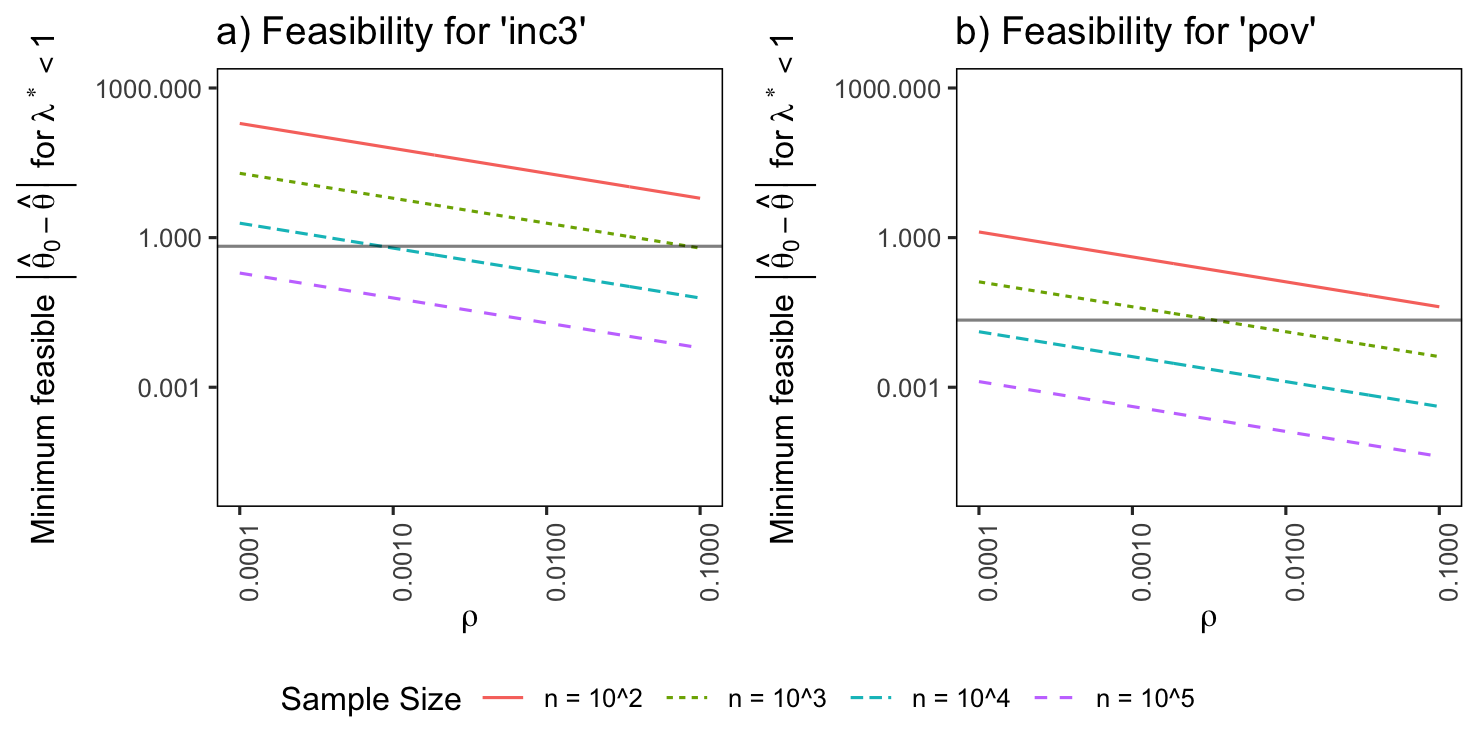}
    \caption{Theoretical minimum AWD for which $\lambda^* < 1$ (y-axis) i.e., survey weighting design is not ignorable under $\rho$-zCDP, as a function of sample size $n$ (colored lines) and privacy loss budget $\rho$. Subplots and horizontal dashed lines refer to realized AWDs for two variables: \texttt{inc3} (left) and \texttt{pov} (right).}
    \label{fig:psid_feas}
\end{figure}

\begin{figure}
    \centering
    \includegraphics[width=.99\textwidth]{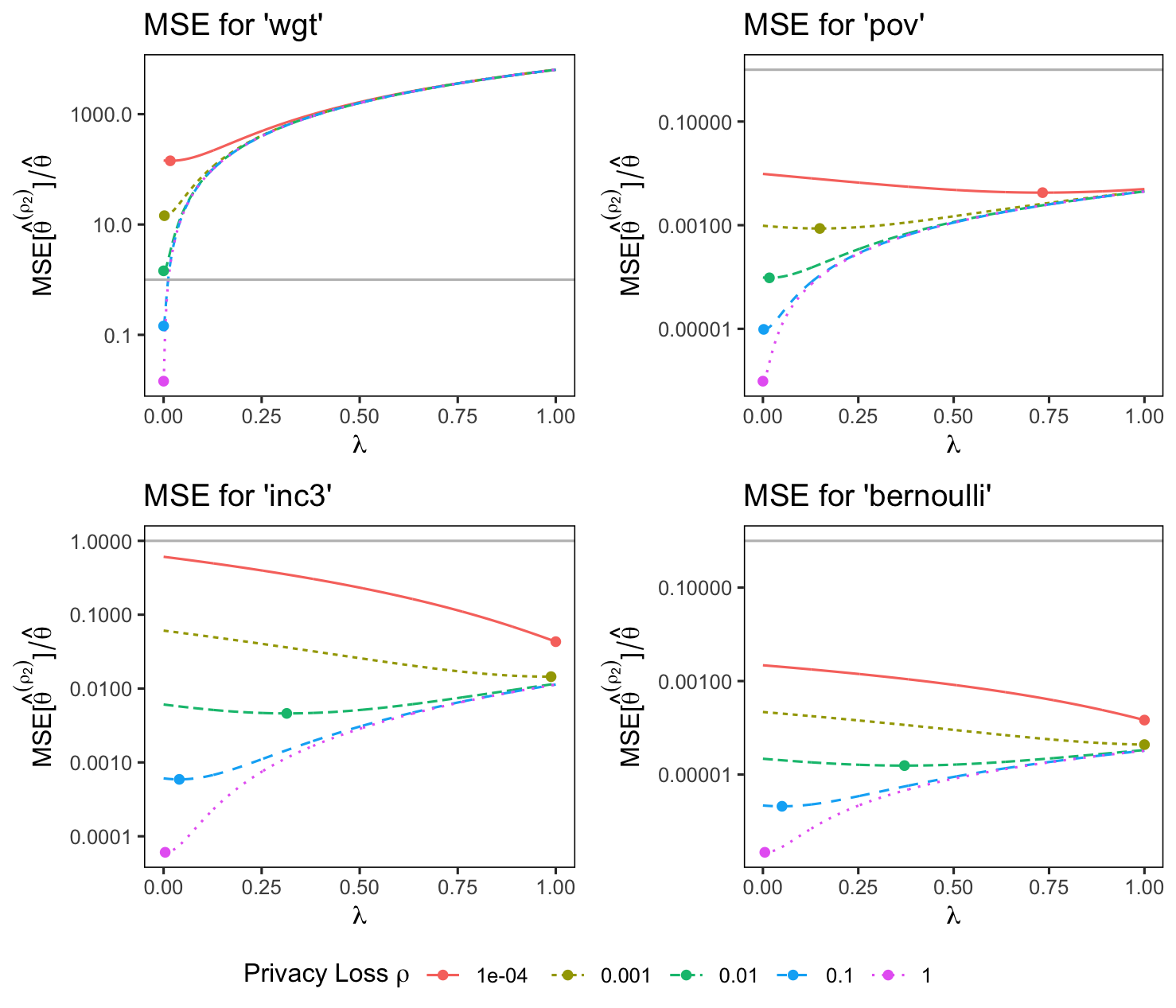}
    \caption{Realized noise-to-signal (DP mean square error divided by non-DP mean estimate, y-axis) as a function of $\lambda$ (x-axis) for different values of privacy loss budget $\rho_2$ (colored lines). Subplots are ordered with decreasing correlation between response variable and survey weights. Points refer to theoretical minimum values, which depend on confidential data and do not satisfy DP.}
    \label{fig:psid_mse_by_var}
\end{figure}

Figure \ref{fig:psid_mse_by_var} shows how the theoretical bias-precision-privacy trade-off manifests for estimating the survey-weighted average cube-root-transformed income (\texttt{inc3}) and proportion of families below the 2019 poverty line (\texttt{pov}). For comparison purposes, the figure also includes the simulated Bernoulli variable 
(\texttt{bern}) and a copy of the survey weights themselves (\texttt{wgt}), representing minimal and maximal correlation between survey weights and responses. We plot the noise-to-signal ratio as the theoretical MSE over the weighted mean estimate on the $y$-axis, with the regularization parameter $\lambda$ on the $x$-axis. We see that as the magnitude of the bias decreases (moving from top left subfigure to bottom right subfigure), the optimal MSE is achieved at larger values of $\lambda^*$ for the same privacy loss budget $\rho_2$. Moreover, as $\rho_2$ decreases, $\lambda^*$ increases for each response variable under consideration. For reasonably small choices of $\rho_2$, we tend to reject small $\lambda$ to optimize the bias-precision trade-off at each fixed $\rho_2$ value.

\begin{figure}
    \centering
    \includegraphics[width=.8\textwidth]{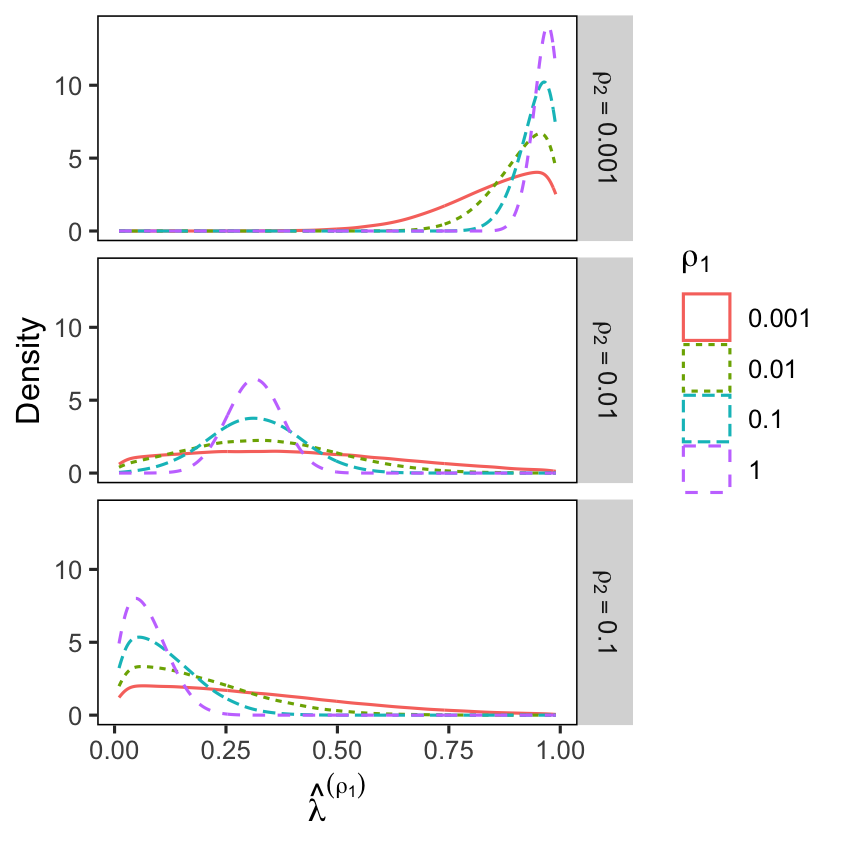}
    \caption{Kernel density estimates for the distribution of $\hat{\lambda}^{(\rho_1)}$ for different values of $\rho_1$ (colored density plot lines) and $\rho_2$ (subfigures) for \texttt{inc3}. Black dashed vertical lines refer to the confidential $\lambda^*$ for each subplot.}
    \label{fig:psid_inc3_lambda_sampling}
\end{figure}

Figure \ref{fig:psid_inc3_lambda_sampling} displays the sampling distribution of $\hat{\lambda}^{(\rho_1)}$ for estimating the average cube-root-transformed income at different values of $\rho_1$ and $\rho_2$. 
This statistic is the most sensitive, as there is a gap between the weighted and unweighted estimates. Therefore, we do not sample $\lambda$ particularly close to the optimal $\hat{\lambda}$ without a large $\rho_1$. However, even for small values of $\rho_1$, we generally avoid sampling small values of $\lambda$ with high probability, which allows us to avoid the worst of the sensitivity inflation in the next stage. 

\begin{figure}
    \centering
    \includegraphics[width=.8\textwidth]{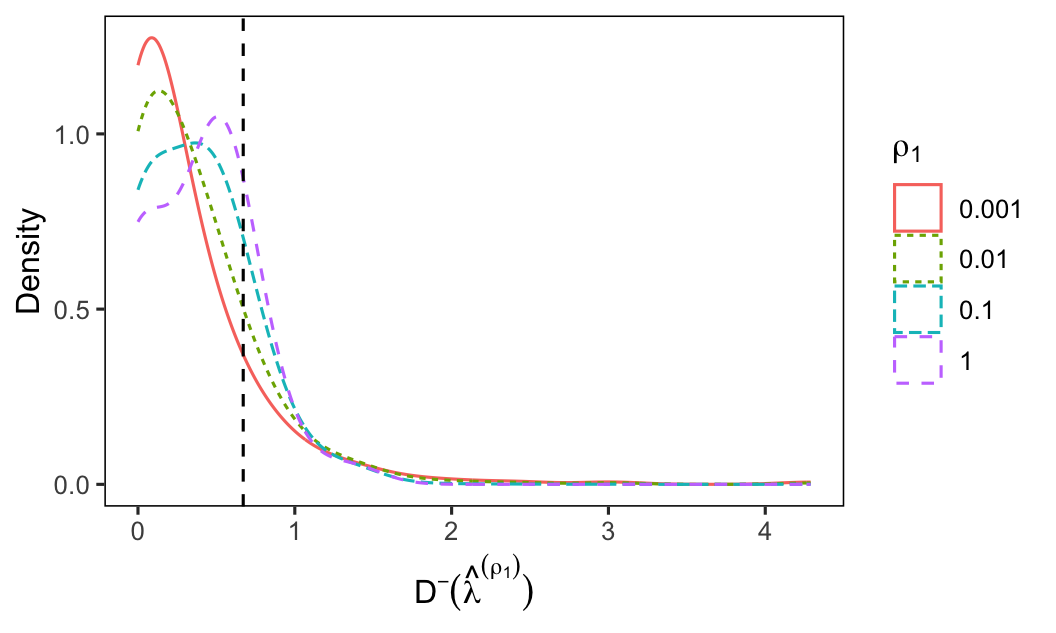}
    \caption{Kernel density estimate for the distribution of the plug-in estimates for the plug-in bias estimate $D^{-}(\hat{\lambda}^{(\rho_1)})$ for different values of $\rho_1$ (colored density plot lines) for \texttt{inc3}. Black dashed vertical line refers to the bias values based on the confidential data.}
    \label{fig:psid_inc3_lambda_bias_recov}
\end{figure}

Figure \ref{fig:psid_inc3_lambda_bias_recov} displays the sampling distribution of the plug-in bias estimate $D^{-}(\hat{\lambda}^{(\rho_1)})$ at different levels of $\rho_1$. As expected, when $\rho_1$ is small, we are more likely to sample larger $\hat{\lambda}^{(\rho_1)}$ values which implies we underestimate the magnitude of the bias. As expected, however, the plug-in estimate improves as $\rho_1$ increases.

\subsection{End-to-end DP Inferences}

In this section, we simulate DP confidence intervals using Algorithm \ref{alg:priv_reg_est} and Algorithm \ref{alg:priv_reg_ci} for the survey weighted population mean of \texttt{inc3}, assessing their width and coverage properties. We consider $\rho_1, \rho_2, \rho_3 \in [10^{-3}, 10^{-1}]$, which covers the full spectrum of regularization from $\lambda^*$, as shown in Figure \ref{fig:psid_mse_by_var}. We also vary $\alpha_v$, i.e., the $(1 - \alpha_v)$*100\% interval width upper bound, to show trade-offs between coverage and interval width.

\begin{figure}
    \centering
    \includegraphics[width=.9\textwidth]{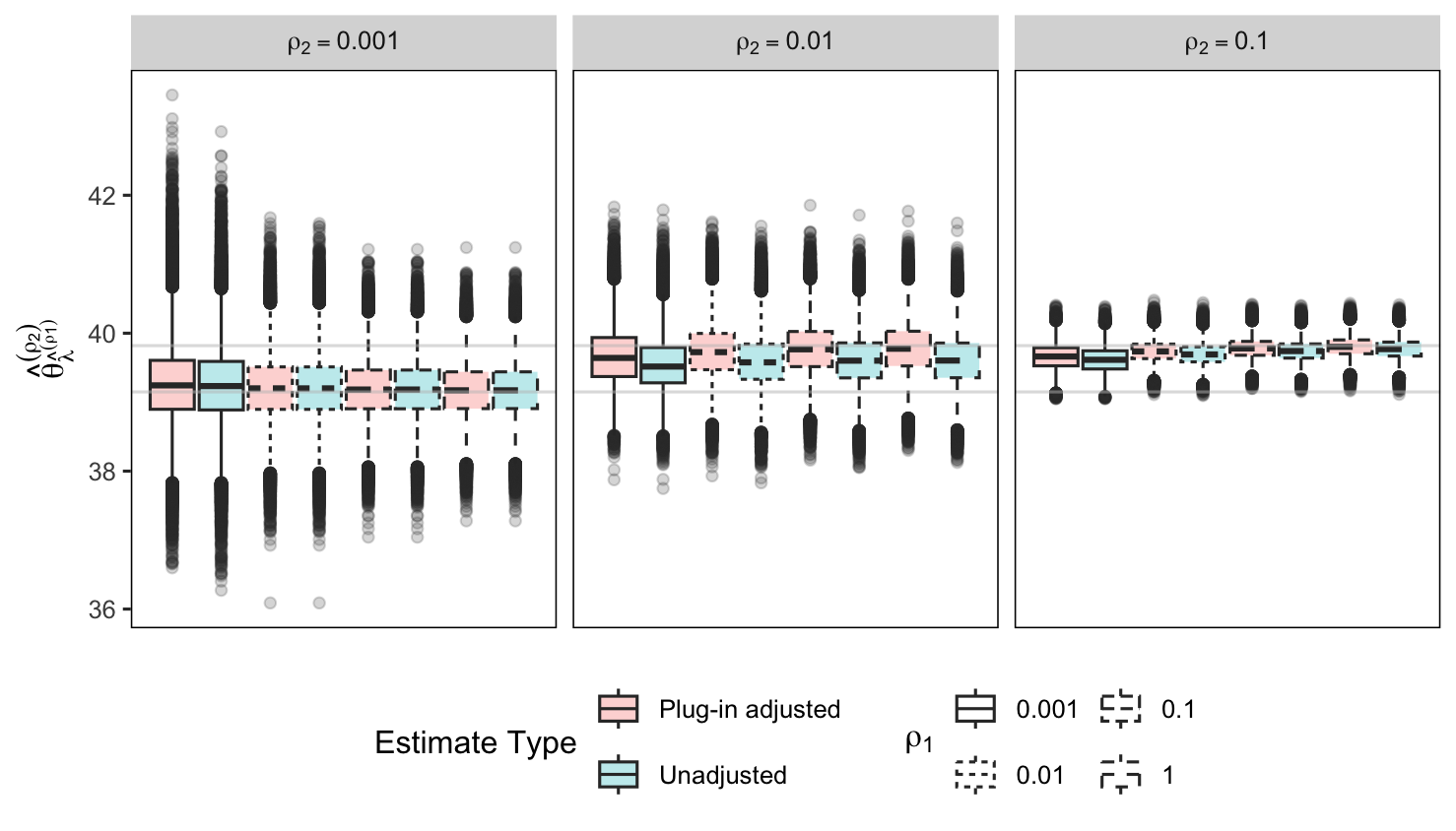}
    \caption{Boxplots of empirical simulations from Algorithm \ref{alg:priv_reg_est} for \texttt{inc3} at different values of $\rho_1$ (color) and $\rho_2$ (subfigures). Boxplot outlines refer to the raw estimate (dashed) and the plug-in bias-adjusted estimates (solid). Green dashed line refers to the confidential weighted mean estimate, and the red dashed line refers to the confidential unweighted mean estimate. }
    \label{fig:e2e_dp_samps}
\end{figure}

Figure \ref{fig:e2e_dp_samps} displays boxplots of samples for $\hat{\theta}_{\hat{\lambda}^{(\rho_1)}}^{(\rho_2)}$ at different values of $\rho_1$ and $\rho_2$. The red and green dashed lines refer to the unweighted and weighted non-DP estimates, respectively. As $\rho_2$ increases (subplots), we are able to accommodate greater survey weighting corrections relative to the additive noise magnitude. Moreover, for smaller values of $\rho_2$, spending more on $\rho_1$ (colored boxplots) reduces the overall variability of the point estimate.

\begin{figure}
    \centering
    \includegraphics[width=.99\textwidth]{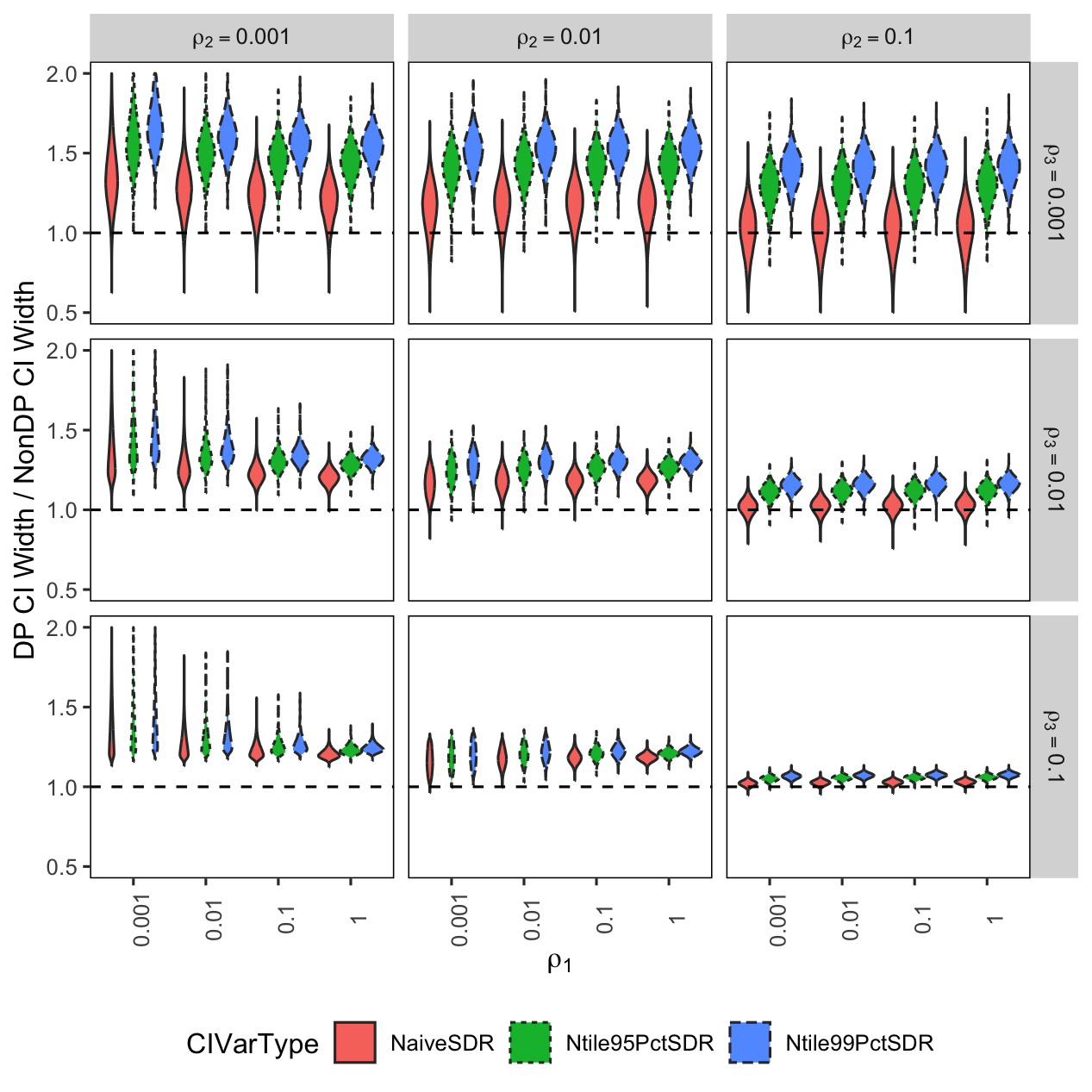}
    \caption{Ratio of DP to non-DP confidence interval widths (y-axis) by values of $\rho_1$ (x-axis), $\rho_2$ (subplot columns), $\rho_3$ (subplot rows), and $\alpha_v$ (colors). Dashed line corresponds to equality (1:1 ratio). }
    \label{fig:ci_widths}
\end{figure}

Figure \ref{fig:ci_widths} compares the interval widths of our proposed algorithm to their non-DP counterparts. For each violin plot, we show the distribution of the ratio for the DP confidence interval over the non-DP confidence interval. The dashed horizontal line at 1.0 corresponds to equality. As expected, increasing either $\rho_1$, $\rho_2$, or $\rho_3$ decreases the DP confidence interval width relative to the non-DP interval width. Of particular interest is different values for $\alpha_v$, represented by the different violin plot colors (.5, .05, and .01, respectively). As expected, decreasing $\alpha_v$ gives us wider confidence intervals by accounting for more potential uncertainty in the interval width. 

\begin{figure}
    \centering
    \includegraphics[width=.99\textwidth]{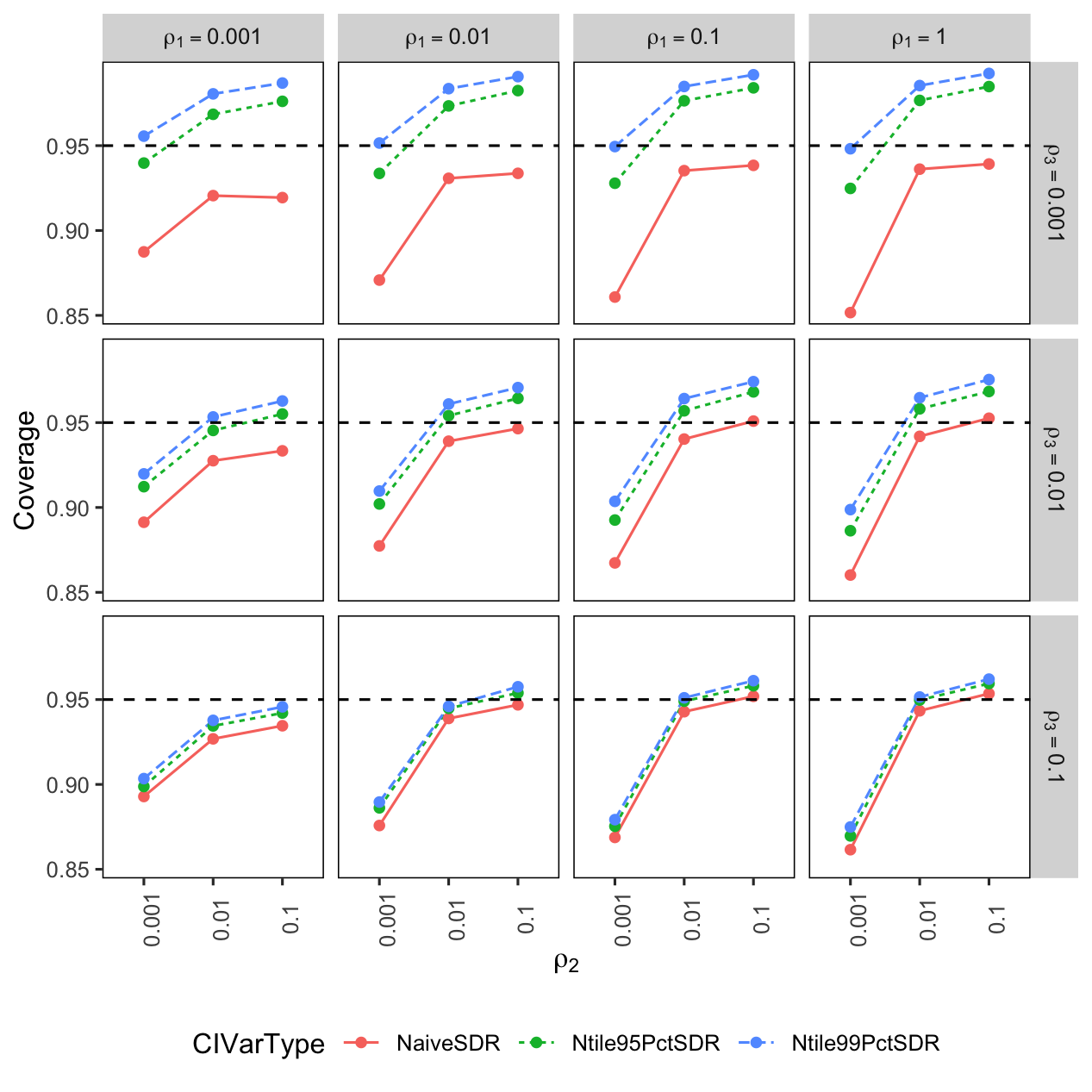}
    \caption{Empirical 95\% confidence interval coverage from simulated ground truth population estimates (y-axis) by values of $\rho_1$ (subplot columns), $\rho_2$ (x-axis), $\rho_3$ (subplot rows), and $\alpha_v$ (colors). Dashed line corresponds to 95\% coverage, the intended target.}
    \label{fig:coverage95}
\end{figure}

Figure \ref{fig:coverage95} estimates the average empirical coverage for 95\% confidence intervals as a function of $\rho_1, \rho_2, \rho_3$, and $\alpha_v$. To do this, we simulate a ``true" parameter $\theta$ from the non-DP normal approximation for each estimate sample, and we empirically average the proportion of intervals covering this simulated ground truth. The dashed line represents 95\% coverage, the intended target. As expected, as all the $\rho$ values increase, the empirical coverage tends towards the equivalent non-DP coverage. Similarly, decreasing $\alpha_v$ increases the empirical coverage probability; in particular, when $\rho_3$ decreases, more conservative values of $\alpha_v$ admit better coverage, as expected.

\section{Discussion}
\label{sec:discussion}
This paper theoretically and empirically suggests that survey weight regularization, when used appropriately, can  reduce the amount of additional noise needed to preserve DP. By adaptively considering how much to shrink weights toward uniform values while satisfying DP, we develop methods that are operationally feasible while facilitating  uncertainty quantification at different precision levels throughout the entire estimation process. 

While our proposed methods can admit the construction of valid finite-sample confidence intervals and asymptotic confidence intervals, different parameter choices may produce intervals that are either too 
wide or 
narrow in practice. When selecting privacy loss budgets for each stage of the algorithm, we recommend incorporating as much domain knowledge as possible. For example, by simulating a distribution of plausible AWD values from prior knowledge, one can establish which kinds of survey weighting biases could be correctable 
at different privacy loss budgets without peeking at the confidential data.

While DP theoretically forbids using data-dependent hyperparameters without DP mechanisms, many commonly used DP algorithms and analyses do not adhere to this rule \cite{abadi_deep_2016,tramer_adversarial_2019}, necessarily yielding additional privacy vulnerabilities in practice \cite{papernot_hyperparameter_2021,mohapatra_role_2022}. It could be the case that tuning certain hyperparameters could substantially improve the end-to-end usefulness of our estimators at a modest expense to privacy risk. 
Understanding this would require a much more extensive and nuanced privacy analysis than 
a simple comparison of privacy loss budgets. Still, such 
privacy analysis could help illuminate 
where DP itself fundamentally limits the kinds of statistical validity offered in survey settings, where worst-case data generating scenarios may be unrealistic in practice. 

\subsection*{Disclosure Statement}
The authors have no conflicts of interest to declare.

\subsection*{Acknowledgments}
Thanks to the National Bureau of Economic Research (NBER) and the Michigan Institute for Data and AI in Society (MIDAS) at the University of Michigan for supporting this research.

\bibliographystyle{alpha}
\bibliography{latex/references}

\appendix


\section{Proofs}
\label{apx:proofs}

\subsection{Proof of Lemma \ref{lem:tri}} 

Define the equivalent $\rho$-zCDP estimator using $\hat{\theta}_{\lambda}$, where
\begin{equation}
    \label{eq:gmech_reg}
    \hat{\theta}^{(\rho )}_{\lambda} \triangleq \hat{\theta}_{\lambda}(\bm{y}, \bm{w}) + \varepsilon, \qquad \varepsilon \sim N\left(0, \frac{\Delta(\hat{\theta})^2}{2\rho} \right).
\end{equation}
We then introduce a loss function that accounts for sources of error introduced by DP relative to the non-DP statistic $\hat{\theta}$ at a fixed $\lambda$ and privacy loss parameter $\rho$ used in applying the Gaussian mechanism, i.e.
\begin{align}
    \ell(\lambda; \bm{y}, \bm{w}) &\triangleq \Ex_{\varepsilon}\left[ (\hat{\theta}_{\lambda}^{(\rho)} -\hat{\theta})^2 \right] \\
    &= \Ex_{\varepsilon}\left[ (\hat{\theta}_{\lambda}^{(\rho)} - \hat{\theta}_{\lambda} + \hat{\theta}_{\lambda} -\hat{\theta})^2 \right] \\
    &= \Ex_{\varepsilon}\left[ (\hat{\theta}_{\lambda}^{(\rho)} - \hat{\theta}_{\lambda})^2 \right] + (\hat{\theta}_{\lambda} -\hat{\theta})^2 + 2 \Ex_{\varepsilon}\left[ \hat{\theta}_{\lambda}^{(\rho)} - \hat{\theta}_{\lambda} \right] \left( \hat{\theta}_{\lambda} -\hat{\theta}\right) \\
    &= \frac{\Delta(\hat{\theta}_{\lambda})^2}{2\rho} + D(\lambda)^2.
\end{align}
Unpacking this statement yields
\begin{align}
    \ell(\lambda; \bm{y}, \bm{w}) &= \label{eq:reg_loss} \frac{1}{2\rho}\left( \frac{\Delta_Y}{N} \right)^2 \left((1 - \lambda) U_W + \frac{\lambda N}{n} \right)^2 + \lambda^2 (\hat{\theta}_0 - \hat{\theta})^2 \\
    &= \frac{1}{2\rho} \left( \frac{\Delta_Y}{N} \right)^2\left(U_W + \lambda \left(\frac{N}{n} - U_W \right) \right)^2 + \lambda^2 (\hat{\theta}_0 - \hat{\theta})^2 \\
    &= \frac{1}{2\rho} \left( \frac{\Delta_Y}{N} \right)^2 \left[U_W^2 + \lambda^2 \left(\frac{N}{n} - U_W \right)^2 + 2 \lambda U_W \left(\frac{N}{n} - U_W \right) \right] + \lambda^2 (\hat{\theta}_0 - \hat{\theta})^2.
\end{align}
This implies
\begin{align}
    \frac{\partial \ell}{\partial \lambda} = \lambda \left[ \frac{1}{\rho} \left( \frac{\Delta_Y}{N} \right)^2 \left(\frac{N}{n} - U_W \right)^2 + 2 (\hat{\theta}_0 - \hat{\theta})^2 \right] + \frac{U_W}{\rho} \left( \frac{\Delta_Y}{N} \right)^2 \left(\frac{N}{n} - U_W \right).
\end{align}
Setting this to zero yields the critical value
\begin{align}
    \label{eq:lambda_crit}
    \lambda^* \triangleq \frac{\frac{U_W}{\rho} \left( \frac{\Delta_Y}{N} \right)^2 \left(U_W - \frac{N}{n} \right)}{\left[ \frac{1}{\rho} \left( \frac{\Delta_Y}{N} \right)^2 \left(U_W - \frac{N}{n} \right)^2 + 2 (\hat{\theta}_0 - \hat{\theta})^2 \right]}.
\end{align}

Since $U_W > \frac{N}{n}$ under informative (i.e., non-uniform sampling weights), which is true by construction under our data generating assumptions, we guarantee that $\lambda^* > 0$. Similarly, we see that
\begin{align}
    &\lambda^* \leq 1 \\
    &\iff \frac{1}{\rho} \left( \frac{\Delta_Y}{N} \right)^2 \left( U_W\left(U_W - \frac{N}{n} \right) \right) \leq \frac{1}{\rho} \left( \frac{\Delta_Y}{N} \right)^2 \left(U_W - \frac{N}{n} \right)^2 + 2 (\hat{\theta}_0 - \hat{\theta})^2 \\
    &\iff - 2 (\hat{\theta}_0 - \hat{\theta})^2 \leq \frac{1}{\rho} \left( \frac{\Delta_Y}{N} \right)^2 \left(U_W - \frac{N}{n} \right)^2 - \frac{1}{\rho} \left( \frac{\Delta_Y}{N} \right)^2 \left( U_W\left(U_W - \frac{N}{n} \right) \right)  \\
    &\iff - 2 \rho \left( \frac{N}{\Delta_Y} \right)^2 (\hat{\theta}_0 - \hat{\theta})^2 \leq \left(U_W - \frac{N}{n} \right)^2 - U_W\left(U_W - \frac{N}{n} \right) \\
    &\iff - 2 \rho \left( \frac{N}{\Delta_Y} \right)^2 (\hat{\theta}_0 - \hat{\theta})^2 \leq U_W^2 + \frac{N^2}{n^2} - \frac{2U_W N}{n} - U_W^2 + \frac{U_W N}{n} \\
    &\iff - 2 \rho \left( \frac{N}{\Delta_Y} \right)^2 (\hat{\theta}_0 - \hat{\theta})^2 \leq  \frac{N^2}{n^2} - \frac{U_W N}{n}  \\
    &\iff - 2 \rho \left( \frac{N}{\Delta_Y} \right)^2 (\hat{\theta}_0 - \hat{\theta})^2 \leq  \frac{N}{n} \left(\frac{N}{n} - U_W \right)  \\
    &\iff  (\hat{\theta}_0 - \hat{\theta})^2 \geq  \frac{\Delta_Y^2}{2\rho N^2} \left( \frac{N}{n} \right) \left(U_W - \frac{N}{n} \right)  \\
    &\iff \left| \hat{\theta}_0 - \hat{\theta} \right| \geq \sqrt{ \frac{\Delta_Y^2}{2\rho N n} \left(U_W - \frac{N}{n} \right) }.  
\end{align}

\subsection{Proof of Theorem 1}

\begin{theorem}
    Algorithm \ref{alg:priv_reg_est} satisfies $\left(\sum_{i=1}^3 \rho_i \right)$-zCDP. 
\end{theorem}
\begin{proof}
    From our previous results, sampling from  \ref{eq:alg1} satisfies $\rho_1$-zCDP, and  sampling from Equations \ref{eq:alg2} and \ref{eq:alg3} satisfy $\rho_2$-zCDP and $\rho_3$-zCDP, respectively. Let these release mechanisms be defined as randomized algorithms $M_1(\bm{y}, \bm{w}), M_2(M_1(\bm{y}, \bm{w}), \bm{y}, \bm{w}),$ and $M_3(M_1(\bm{y}, \bm{w}), \bm{y}, \bm{w})$. By adaptive composition, releasing the joint vector
    \begin{align}
        \begin{pmatrix}
            \hat{\lambda}^{(\rho_1)} \\ \hat{\theta}_{\hat{\lambda}^{(\rho_1)}}^{(\rho_2)} \\ \widehat{\Var}_{\hat{\lambda}^{(\rho_1)}}(\hat{\theta})^{\rho_3}
        \end{pmatrix} = \begin{pmatrix}
            M_1(\bm{y}, \bm{w}) \\
            M_2(M_1(\bm{y}, \bm{w}), \bm{y}, \bm{w}) \\
            M_3(M_1(\bm{y}, \bm{w}), \bm{y}, \bm{w})
        \end{pmatrix}
    \end{align}
    satisfies $\left( \sum_{i=1}^3 \rho_i \right)$-zCDP. 

However, this statistic depends on the confidential data and is as sensitive as the original survey-weighted mean. We see that

\begin{align}
    \Delta(|\hat{\theta}_0 - \hat{\theta}|) &= \sup_{(\bm{y}, \bm{w}) \sim_M (\bm{y}', \bm{w}')} \left| |\hat{\theta}_0(\bm{y}, \bm{w}) - \hat{\theta}(\bm{y}, \bm{w})| - |\hat{\theta}_0(\bm{y}', \bm{w}') - \hat{\theta}(\bm{y}', \bm{w}')| \right| \\
    &\leq \sup_{(\bm{y}, \bm{w}) \sim_M (\bm{y}', \bm{w}')} \left| \hat{\theta}_0(\bm{y}, \bm{w}) - \hat{\theta}(\bm{y}, \bm{w}) - \hat{\theta}_0(\bm{y}', \bm{w}') + \hat{\theta}(\bm{y}', \bm{w}') \right| \\
    &= \sup_{(\bm{y}, \bm{w}) \sim_M (\bm{y}', \bm{w}')} \left| \left(\frac{ y_1 - y_1'}{n} \right) - \left( \frac{y_1' w_1' - y_1 w_1}{N} \right) \right| \\
    &\leq \left| \left(\frac{U_W U_Y - L_W L_Y}{N} \right) - \frac{\Delta_Y}{n} \right| \\
    &= \Delta(\hat{\theta}) - \Delta(\hat{\theta}_0) \\
\implies \Delta((\hat{\theta}_0 - \hat{\theta})^2) &= (\Delta(\hat{\theta}) - \Delta(\hat{\theta}_0))^2 \\
\implies \Delta(\ell) &= \sup_{\lambda \in [0, 1]} \sup_{(\bm{y}, \bm{w}) \sim_M (\bm{y}', \bm{w}')} \left| L(\bm{y}, \bm{w}, \lambda) - L(\bm{y}', \bm{w}', \lambda)\right| \\
&= \sup_{\lambda \in [0, 1]} \sup_{(\bm{y}, \bm{w}) \sim_M (\bm{y}', \bm{w}')} \Bigg| \lambda^2 \big( (\hat{\theta}_0(\bm{y}, \bm{w}) - \hat{\theta}(\bm{y}, \bm{w}))^2 \\
&\hspace{150pt} - (\hat{\theta}_0(\bm{y}', \bm{w}') - \hat{\theta}(\bm{y}', \bm{w}'))^2 \big)  \Bigg| \\
&= (\Delta(\hat{\theta}) - \Delta(\hat{\theta}_0))^2.
\end{align}

Using this expression, we can estimate the optimal $\lambda^*$ with a $\rho_1$-zCDP estimate $\hat{\lambda}^{(\rho_1)}$ by spending some privacy loss $\rho_1$ and implementing the exponential mechanism by sampling from the density
\begin{equation}
    \label{app:eq:lambda_expmech_implement}
    f(\lambda) \propto \ind{\lambda \in [0, 1]} \exp\left(-\frac{\sqrt{2\rho_1}}{2\Delta(\ell)} \ell(\lambda; \bm{y}, \bm{w}) \right).
\end{equation}
Note that  \eqref{app:eq:lambda_expmech_implement} specifies a truncated normal distribution centered at $\hat{\lambda^*}$ and bounded by $\lambda \in [0, 1]$.

To combine everything, we have
\begin{align}
    \p(|\hat{\theta}^{(\rho_2)}_{\hat{\lambda}^{(\rho_1)}} - \hat{\theta}| \leq \gamma) &\leq \p( |\hat{\theta}^{(\rho_2)}_{\hat{\lambda}^{(\rho_1)}} - \hat{\theta}_{\hat{\lambda}^{(\rho_1)}}| + |\hat{\theta}_{\hat{\lambda}^{(\rho_1)}} - \hat{\theta}| \leq \gamma) \\
    &\leq \p\left( |\hat{\theta}^{(\rho_2)}_{\hat{\lambda}^{(\rho_1)}} - \hat{\theta}_{\hat{\lambda}^{(\rho_1)}}| \leq \frac{\gamma}{2},  |\hat{\theta}_{\hat{\lambda}^{(\rho_1)}} - \hat{\theta}| \leq \frac{\gamma}{2} \right). 
\end{align}
where the two error sources are independent conditional on observing $\hat{\lambda}^{(\rho_1)}$. This yields the final result. 
\end{proof}

\end{document}